\documentclass{llncs}

\usepackage{multicol}
\usepackage{amsmath}
\usepackage{amsfonts}
\usepackage{amssymb}
\usepackage{graphicx}
\usepackage{epic}
\usepackage{eepic}
\usepackage{epsfig,float}
\usepackage{verbatim}
\usepackage{pdfsync}
\usepackage{gastex}
\usepackage[lowtilde]{url}

\usepackage{xcolor}

\renewcommand{\le}{\leqslant}
\renewcommand{\ge}{\geqslant}

\newcommand{\eps}{\varepsilon}
\newcommand{\emp}{\emptyset}

\newcommand{\Sig}{\Sigma}

\newcommand{\noin}{\noindent}

\newcommand{\bi}{\begin{itemize}}
\newcommand{\ei}{\end{itemize}}
\newcommand{\be}{\begin{enumerate}}
\newcommand{\ee}{\end{enumerate}}
\newcommand{\bd}{\begin{description}}
\newcommand{\ed}{\end{description}}
\newcommand{\bq}{\begin{quote}}
\newcommand{\eq}{\end{quote}}

\newcommand{\tid}{\mbox{{\bf 1}}}

\newcommand{\cD}{{\mathcal D}}

\newcommand{\cT}{{\mathcal T}}
\newcommand{\cW}{{\mathcal W}}

\title{Upper Bounds on Syntactic Complexity of Left and Two-Sided Ideals\thanks{This work was supported by the Natural Sciences and Engineering Research Council of Canada 
grant No.~OGP000087, 
and by Polish NCN grant DEC-2013/09/N/ST6/01194.}
}

\author{Janusz~Brzozowski\inst{1} \and Marek Szyku{\l}a \inst{2}}

\titlerunning{Syntactic Complexity of Ideals}

\authorrunning{J. Brzozowski and M. Szyku{\l}a}   

\institute{David R. Cheriton School of Computer Science, University of Waterloo, \\
Waterloo, ON, Canada N2L 3G1\\
\{{\tt brzozo@uwaterloo.ca}\}
\and
Institute of Computer Science, University of Wroc{\l}aw,\\
Joliot-Curie 15, PL-50-383 Wroc{\l}aw, Poland\\
\{{\tt msz@cs.uni.wroc.pl}\}
}

\begin{document}
\maketitle
\begin{abstract}
We solve two open problems concerning syntactic complexity.
We prove that the cardinality of the syntactic semigroup of a left ideal or a suffix-closed language with $n$ left quotients (that is, with state complexity $n$) is at most $n^{n-1}+n-1$, 
and that of a two-sided ideal or a factor-closed language is
at most $n^{n-2}+(n-2)2^{n-2}+1$.
Since these bounds are known to be reachable, this settles the problems.
\medskip

\noin
\end{abstract}

\section{Introduction}

The \emph{syntactic complexity}~\cite{BrYe11} of a regular language is the size of its syntactic semigroup~\cite{Pin97}. The \emph{transition semigroup} $T$ of a deterministic finite automaton (DFA) $\cD$ is the semigroup of transformations
of the state set of $\cD$ generated by the transformations induced by the input letters of $\cD$.
The transition semigroup of a minimal DFA of a language $L$ is isomorphic to the 
syntactic semigroup of $L$~\cite{Pin97}; hence syntactic complexity is equal to the cardinality of
$T$.

The number $n$ of states of $\cD$ is known as the \emph{state complexity} of the language~\cite{Brz10,Yu01}, and it is the same as the number of left quotients of the language. 
The \emph{syntactic complexity of a class} of regular languages is the maximal syntactic complexity of languages in that class expressed as a function of $n$.

A \emph{right ideal} (respectively, 
\emph{left ideal}, \emph{two-sided ideal}) is a non-empty language $L$ over an alphabet $\Sig$ such that $L=L\Sig^*$ (respectively, $L=\Sig^*L$, $L=\Sig^*L\Sig^*$).
We are interested only in regular ideals;
for reasons why they deserve to be studied see~\cite[Section~1]{BJL13}.
Ideals appear in pattern matching.
For example, if a \emph{text} is a word $w$ over some alphabet $\Sig$, and 
a \emph{pattern} is  an arbitrary language $L$ over $\Sig$, then
an occurrence of a pattern represented by $L$ in text $w$ is a triple $(u,x,v)$ such that $w=uxv$ and $x$ is in~$L$.
Searching text $w$ for words in $L$ is equivalent to looking for prefixes of $w$ that belong to the language $\Sig^*L$, which is the left ideal generated by $L$.

The syntactic complexity of right ideals was proved to be $n^{n-1}$ in~\cite{BrYe11}.
The syntactic complexities of left and two-sided ideals were also examined in~\cite{BrYe11}, where it was shown that
$n^{n-1}+n-1$ and $n^{n-2}+(n-2)2^{n-2}$, respectively, are lower bounds on these complexities, and it was conjectured that
they are also upper bounds.
In this paper we prove these conjectures.

If $w=uxv$ for some $u,v,x\in\Sigma^*$, then  $v$ is a {\em suffix\/} of $w$ and  
 $x$ is a {\em factor\/} of $w$.
A suffix of $w$ is also a factor of $w$.
A~language $L$ is {\it suffix-closed\/} (respectively, \emph{factor-closed}) if $w\in L$ implies that
every suffix (respectively, factor) of $w$ is also in~$L$. 
We are interested only in regular suffix- and factor-closed languages.
Since every left (respectively, two-sided) ideal is the complement of a 
suffix-closed (respectively, factor-closed) language, and syntactic complexity is preserved by complementation, our theorems also apply to suffix- and factor-closed languages, but our proofs are given for  left and two-sided ideals only.
\section{Preliminaries}

The \emph{left quotient} or simply \emph{quotient} of a regular language $L$ by a word $w$ is denoted by $Lw$ and defined by $Lw=\{x\mid wx\in L\}$.
A language is regular if and only if it has a finite number of quotients.
The number of quotients of $L$ is called its \emph{quotient complexity}. 
We denote the set of quotients by $K=\{K_0,\dots,K_{n-1}\}$, where $K_0=L=L\eps$ by convention.
Each quotient $K_i$ can be represented also as $L{w_i}$, where $w_i\in\Sig^*$ is such that
$L{w_i}=K_i$.

A \emph{deterministic finite automaton (DFA)} is a quintuple 
$\cD=(Q, \Sigma, \delta, q_0,F)$, where
$Q$ is a finite non-empty set of \emph{states},
$\Sig$ is a finite non-empty \emph{alphabet}
$\delta\colon Q\times \Sig\to Q$ is the \emph{transition function},
$q_0\in Q$ is the \emph{initial} state, and
$F\subseteq Q$ is the set of \emph{final} states.

The \emph{quotient DFA} of a regular language $L$ with $n$ quotients is defined by
$\cD=(K, \Sigma, \delta, K_0,F)$, where 
$\delta(K_i,w)=K_j$ if and only if ${K_i}w=K_j$, 
and $F=\{K_i\mid \eps \in K_i\}$.
To simplify the notation, we use the set $Q=\{0,\dots,n-1\}$ of subscripts of quotients to denote the states of $\cD$; then $\cD$ is denoted by
$\cD=(Q, \Sigma, \delta, 0,F)$. The quotient corresponding to $q\in Q$ is then $K_q=\{w\mid \delta(q,w)\in F\}$.
The quotient $K_0=L$ is the \emph{initial} quotient. A quotient is \emph{final} if it contains $\eps$.
A state $q$ is \emph{empty} if its quotient $K_q$ is empty.

The quotient DFA of $L$ is isomorphic to each complete minimal DFA of $L$.
The number of states in the quotient DFA of $L$ (the quotient complexity of $L$) is therefore equal to the state complexity of $L$.

In any DFA, each letter $a\in \Sig$ defines a transformation of the set $Q$ of $n$ states.
Let $\cT_{Q}$ be the set of all $n^n$ transformations of $Q$; then $\cT_{Q}$ is a monoid under  composition. 
The \emph{identity} transformation $\tid$ maps each element to itself.
For $k\ge 2$, a transformation (permutation) $t$ of a set $P=\{q_0,q_1,\ldots,q_{k-1}\} \subseteq Q$ is a \emph{$k$-cycle}
if $q_0t=q_1, q_1t=q_2,\ldots,q_{k-2}t=q_{k-1},q_{k-1}t=q_0$. 
A $k$-cycle is denoted by $(q_0,q_1,\ldots,q_{k-1})$.
If a transformation $t$ of $Q$ acts like a $k$-cycle on some $P \subseteq Q$, we say that $t$ has a $k$-cycle.
A~transformation has a \emph{cycle} if it has a $k$-cycle for some $k\ge 2$.
A~2-cycle $(q_0,q_1)$ is called a \emph{transposition}.
A transformation is \emph{constant} if it maps all states to a single state $q$; it is denoted by $(Q\to q)$.
If $w$ is a word of $\Sig^*$, the fact that $w$ induces transformation $t$ is denoted by 
$w\colon t$.
A~transformation mapping $i$ to $q_i$ for $i=0, \dots, n-1$ is sometimes denoted by
$[q_0, \dots,q_{n-1}]$.

\section{Left Ideals}
\subsection{Basic Properties}
Let $Q=\{0,\ldots,n-1\}$, 
let $\cD_n=(Q, \Sigma_\cD, \delta_\cD, 0,F)$ be a minimal DFA, and let $T_n$ be its transition semigroup.
Consider the sequence $(0,0t,0t^2,\dots)$ of states obtained by applying transformation $t\in T_n$ repeatedly, starting with the initial state.
Since $Q$ is finite, there must eventually be a repeated state, that is, there must exist $i$ and $j$ such that $0,0t,\dots,0t^i,0t^{i+1},\dots, 0t^{j-1}$ are distinct, but $0t^j=0t^i$;
the integer $j-i$ is the \emph{period} of $t$.
If the period  is $1$,  $t$ is said to be \emph{initially aperiodic};
then the sequence is $0,0t,\dots, 0t^{j-1}=0t^j$.

\begin{lemma}[\cite{BrYe11}]
\label{lem:aperiodic}
If $\cD_n$ is a DFA of a left ideal, all the transformations in $T_n$ are initially aperiodic, and no state of $\cD_n$ is empty.
\end{lemma}

\begin{remark}[\cite{BJL13}]
\label{rem:left-ideals_xy}
A language $L\subseteq \Sig^*$ is a left ideal if and only if for all $x,y\in\Sig^*$, $Ly\subseteq L{xy}$.
Hence, if $Lx\neq L$, then $L\subset Lx$ for any $x\in\Sig^+$.
\end{remark}

It is useful to restate this observation it terms of the states of $\cD_n$.
For DFA $\cD_n$ and states $p,q \in Q$, we write $p \prec q$ if $K_p \subset K_q$. 

\begin{remark}\label{rem:left-ideals_xy2}
A DFA $\cD_n$ is a minimal DFA of a left ideal if and only if for all $s,t\in T_n \cup \{\tid\}$, $0t\preceq 0st$.
If $0t\neq 0$, then $0\prec 0t$ for any $t\in T_n$.
Also, if $r\in Q$ has a $t$-predecessor, that is, if there exists $q\in Q$ such that $qt=r$, then $0t\preceq r$. 
(This follows because $q=0s$ for some transformation $s$ since $q$ is reachable from 0; hence $0\preceq q$ and $0t \preceq qt =r$.)
In particular, if $r$ appears in a cycle of $t$ or is a fixed point of $t$, then $0t\preceq r$.
\end{remark}

We consider chains of the form $K_{i_1}\subset K_{i_2}\subset\dots \subset K_{i_h}$,
where the $K_{i_j}$ are quotients of $L$. 
If $L$ is a left ideal, the smallest element of any maximal-length chain is always $L$.
Alternatively, we consider chains of states starting from 0 and strictly ordered by $\prec$.

\begin{proposition}\label{prop:chain}
For $t\in T_n$ and $p, q\in Q$, $p \prec q$ implies $pt \preceq qt$.
If $p\prec pt$, then $p \prec pt \prec \dots \prec pt^k = pt^{k+1}$ for some $k \ge 1$. 
Similarly, $p \succ q$ implies $pt \succeq qt$, and $p \succ pt$ implies $p \succ pt \succ \dots \succ pt^k = pt^{k+1}$ for some  $k \ge 1$.
\end{proposition}

It was proved in~\cite[Theorem 4, p. 124]{BrYe11} that the transition semigroup of the following DFA of a left ideal meets the bound $n^{n-1}+n-1$.

\begin{definition}[Witness: Left Ideals]
\label{def:wit}
For $n\ge 3$, 
we define the DFA 
$\cW_n =(Q,\Sig_\cW,\delta_\cW,0,\{n-1\}),$
where $Q=\{0,\ldots,n-1\}$, $\Sig_\cW=\{a,b,c,d,e\}$, 
and $\delta_\cW$ is defined by $a\colon (1,\ldots,n-1)$,
$b\colon (1,2)$,
$c\colon (n-1\to 1)$,
$d\colon (n-1\to 0)$,
and $e\colon (Q \to 1)$.
For $n=3$, $a$ and $b$ coincide, and we can use $\Sig_\cW=\{b,c,d,e\}$.
\end{definition}

\begin{remark}
\label{rem:left-ideals_transformations}
In $\cW_n$, the transformations induced by $a$, $b$, and $c$ restricted to $Q\setminus \{0\}$ generate all the transformations of the last $n-1$ states. Together with the transformation of $d$, they generate all transformations of $Q$ that fix 0.
To see this, consider any transformation $t$ that fixes 0. If some states from $\{1,\dots,n-1\}$ are mapped to 0 by $t$, we can map them first to $n-1$ and $n-1$ to one of them by the transformations of $a$, $b$, and $c$, and then map $n-1$ to 0 by the transformation of $d$.
Also the words of the form $e a^i$ for $i \in \{0,\ldots,n-2\}$ induce constant transformations $(Q \to i+1)$. Hence  the transition semigroup of $\cW_n$ contains all the constant transformations.
\end{remark}

\begin{example}
One verifies that the maximal-length chains of quotients in  $\cW_n$ have length~2.
On the other hand, for $n\ge 2$, let $\Sig=\{a,b\}$ and let $L=\Sig^*a^{n-1}$.
Then $L$ has $n$ quotients and the maximal-length chains are of length $n$.
\end{example}

\subsection{Upper Bound}

Our main result of this section shows that the lower bound $n^{n-1}+ n-1$ is also an upper bound.
Our approach is as follows: We consider a minimal DFA $\cD_n=(Q, \Sigma_\cD, \delta_\cD, 0,F)$, where $Q=\{0,\ldots,n-1\}$, of an arbitrary left ideal with $n$ quotients and let $T_n$ be the transition semigroup of $\cD_n$. 
We also deal  with the witness DFA  $\cW_n =(Q,\Sig_\cW,\delta_\cW,0,\{n-1\})$ of Definition~\ref{def:wit} that has the same state set as $\cD_n$ and whose transition semigroup is $S_n$. We shall show that there is an injective mapping $f\colon T_n\to S_n$, and this will prove that $|T_n|\le |S_n|$.

\begin{remark}
\label{rem:smalln}
If $n=1$, the only left ideal is $\Sig^*$ and the transition semigroup of its minimal DFA satisfies the bound
$1^0+1-1=1$.
If $n=2$, there are only three allowed transformations, since the transposition
$(0,1)$ is not initially aperiodic and so is ruled out by Lemma~\ref{lem:aperiodic}.
Thus the bound $2^1+2-1=3$ holds.
\end{remark}


\begin{lemma}
\label{lem:chain2}
If $n\ge 3$ and a maximal-length chain in $\cD_n$ strictly ordered by $\prec$ has length 2, then $|T_n|\le n^{n-1}+n-1$ and $T_n$ is a subsemigroup of $S_n$.

\end{lemma}
\begin{proof}
Consider an arbitrary transformation $t\in T_n$ and let $p=0t$. 
If $p=0$, then any state other than $0$ can possibly be mapped by $t$  to any one of the $n$ states; hence there are at most $n^{n-1}$ such transformations. All of these transformations are in $S_n$ by Remark~\ref{rem:left-ideals_transformations}.

If $p\neq 0$, then $0\prec p$. Consider any state $q\not \in \{0,p\}$; by Remark~\ref{rem:left-ideals_xy2}, $p\preceq qt$.
If $p \neq qt$, then $p\prec qt$.
But then we have the chain $0\prec p\prec qt$ of length 3, contradicting our assumption.
Hence we must have $p=qt$, and so $t$ is the constant transformation $t= (Q\to p)$.
Since $p$ can be any one of the $n-1$ states other than 0, we have at most $n-1$ such transformations. 
Since all of these transformations are in $S_n$ by Remark~\ref{rem:left-ideals_transformations}, $T_n$ is a subsemigroup of $S_n$.
\qed
\end{proof}

\begin{theorem}[Left Ideals, Suffix-Closed Languages]\label{thm:left-ideals_upper-bound}
If $n\ge 3$ and $L$ is a left ideal or a suffix-closed language with $n$ quotients, then its syntactic complexity is less than or equal to 
$n^{n-1}+n-1$.
\end{theorem}
\begin{proof}
It suffices to prove the result for left ideals.
For a transformation $t \in T_n$, consider the following cases:
\goodbreak
\smallskip

\noindent
\hglue 15pt  
{\bf Case 1:} $t \in S_n$. \\
Let $f(t) = t$; obviously $f(t)$ is injective.
\smallskip

\noindent
\hglue 15pt  
{\bf Case 2:} $t \not\in S_n$ and $0t^2 \neq 0t$.\\
Note that $t \not\in S_n$ implies $0 t \neq 0$ by Remark~\ref{rem:left-ideals_transformations}.
Let $0t=p$.
We have $p=0t \prec 0tt=pt$ by Remark~\ref{rem:left-ideals_xy2}. Let $p \prec \dots \prec p t^k = p t^{k+1}$ be the chain defined from $p$; 
this chain is of length at least 2.
Let $f(t)=s$, where $s$ is the transformation defined by

\begin{center} 
$0 s = 0, \quad p t^k s = p, \quad q s = q t \text{ for the other states } q\in Q.$
\end{center}
Transformation $s$ is shown in Figure~\ref{fig:left-case2}, where the dashed transitions show how $s$ differs from $t$.
\begin{figure}[ht]
\unitlength 10pt
\begin{center}\begin{picture}(26,12)(0,0)
\gasset{Nh=2.5,Nw=2.5,Nmr=1.25,ELdist=0.5,loopdiam=2}
\node[Nframe=n](name)(0,12){$t\colon$}
\node(0)(2,10){0}
\node(p)(8,10){$p$}
\node(pt)(14,10){$pt$}
\node[Nframe=n](pdots)(20,10){$\dots$}
\node(pt^k)(26,10){$pt^k$}
\drawedge(0,p){$t$}
\drawedge(p,pt){$t$}
\drawedge(pt,pdots){$t$}
\drawedge(pdots,pt^k){$t$}
\drawloop[loopangle=90](pt^k){$t$}

\node[Nframe=n](name)(0,4){$s\colon$}
\node(0')(2,2){0}
\node(p')(8,2){$p$}
\node(pt')(14,2){$pt$}
\node[Nframe=n](pdots')(20,2){$\dots$}
\node(pt^k')(26,2){$pt^k$}
\drawloop[loopangle=90,dash={.5 .25}{.25}](0'){$s$}
\drawedge(p',pt'){$s$}
\drawedge(pt',pdots'){$s$}
\drawedge(pdots',pt^k'){$s$}
\drawedge[curvedepth=2,dash={.5 .25}{.25}](pt^k',p'){$s$}
\end{picture}\end{center}
\caption{Case~2 in the proof of Theorem~\ref{thm:left-ideals_upper-bound}.}
\label{fig:left-case2}
\end{figure}
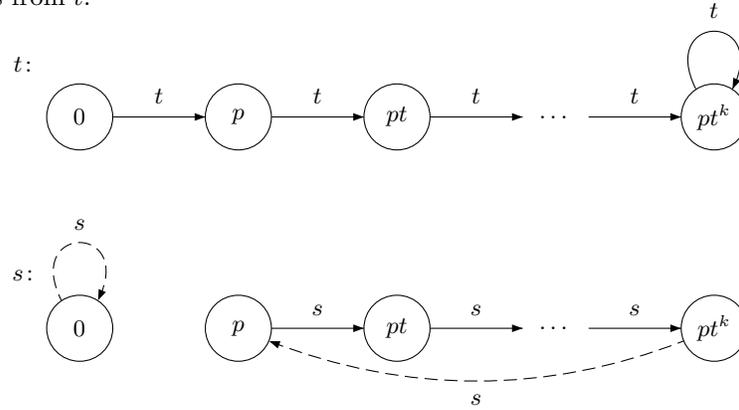

By Remark~\ref{rem:left-ideals_transformations}, $s \in S_n$.
However, $s\not\in T_n$, as it contains the cycle $(p, \ldots, p t^k)$ with states strictly ordered by $\prec$ in DFA $\cD_n$, which contradicts Proposition~\ref{prop:chain}.
Since $s \not\in T_n$, it is distinct from the transformations defined in Case~1.

In going from $t$ to $s$, we have added one transition ($0s=0$) that is a fixed point, and one ($pt^ks=p$) that is not.
Since only one non-fixed-point transition has been added, there can be only one cycle in $s$ with states strictly ordered by $\prec$.
Since $0$ can't appear in this cycle, $p$ is its smallest element with respect to $\prec$.

Suppose now that $t'\neq t$ is another transformation that satisfies Case~2, that is, $0 t' = p' \neq 0$ and $p' t' \neq p'$; 
we shall show that $f(t) \neq f(t')$.
Define $s'$ for $t'$ as $s$ was defined for $t$.
For a contradiction, assume $s = f(t) = f(t')=s'$. 

Like $s$, $s'$ contains only one cycle strictly ordered by $\prec$, and $p'$ is its smallest element.
Since we have assumed that $s=s'$, we must have $p=0t=0t'=p'$ and the cycles in $s$ and $s'$ must be identical. 
In particular, $p t^k t = p t^k = p (t')^k t' = p (t')^k$.
For $q$ of $Q\setminus \{0,pt^k\}$, we have $qt=qs=qs'=qt'$.  
Hence $t = t'$---a contradiction. 
Therefore $t\neq t'$ implies $f(t)\neq f(t')$.
\smallskip

\noindent
\hglue 15 pt
{\bf Case 3:}  $t \not\in S_n$ and $0t^2 = 0t$. \\
As before, let $0t=p$. Consider any state $q\not\in \{0,p\}$; then
$0\prec q$ by Remark~\ref{rem:left-ideals_xy2} and $0t\preceq qt$ by Proposition~\ref{prop:chain}.
Thus either $p\prec qt$, or $p= qt$.
We consider the following sub-cases:
\smallskip

\noindent
\hglue 15 pt
$\bullet$ {\bf (a):} $t$ has a cycle.\\
Since $t$ has a cycle, take a state $r$ from the cycle; then $r$ and $rt$ are not comparable under $\preceq$ by Proposition~\ref{prop:chain}, and $p \prec r$ by Remark~\ref{rem:left-ideals_xy2}.
Let $f(t) = s$, where $s$ is the transformation shown in Figure~\ref{fig:left-case3a} and defined by
\begin{center}
$0 s = 0, \quad p s = r, \quad q s = q t \text{ for the other states } q \in Q.$
\end{center}
\begin{figure}[ht]
\unitlength 10pt
\begin{center}\begin{picture}(19,13)(0,0)
\gasset{Nh=2.5,Nw=2.5,Nmr=1.25,ELdist=0.5,loopdiam=2}
\node[Nframe=n](name)(0,13){$t\colon$}
\node(0)(2,10){0}
\node(p)(8,10){$p$}
\node(r)(14,10){$r$}
\node(rt)(16.5,13){$rt$}
\node[Nframe=n](rdots)(19,10){$\dots$}
\drawedge(0,p){$t$}
\drawloop[loopangle=90](p){$t$}
\drawedge[curvedepth=1](r,rt){$t$}
\drawedge[curvedepth=1](rt,rdots){$t$}
\drawedge[curvedepth=1](rdots,r){$t$}

\node[Nframe=n](name)(0,4){$s\colon$}
\node(0')(2,2){0}
\node(p')(8,2){$p$}
\node(r')(14,2){$r$}
\node(rt')(16.5,5){$rt$}
\node[Nframe=n](rdots')(19,2){$\dots$}
\drawloop[loopangle=90,dash={.5 .25}{.25}](0'){$s$}
\drawedge[dash={.5 .25}{.25}](p',r'){$s$}
\drawedge[curvedepth=1](r',rt'){$s$}
\drawedge[curvedepth=1](rt',rdots'){$s$}
\drawedge[curvedepth=1](rdots',r'){$s$}
\end{picture}\end{center}
\caption{Case~3(a) in the proof of Theorem~\ref{thm:left-ideals_upper-bound}.}
\label{fig:left-case3a}
\end{figure}
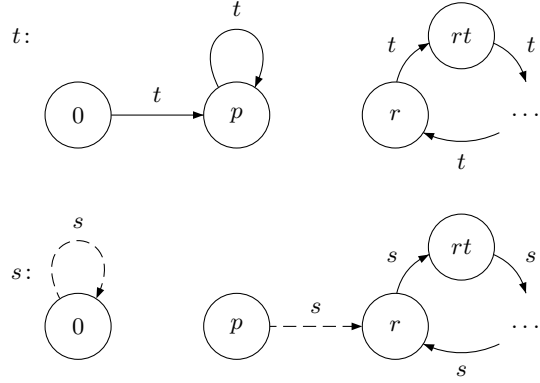

By Remark~\ref{rem:left-ideals_transformations}, $s \in S_n$. 
Suppose that $s \in T_n$; since $p \prec r$, we have 
$r=ps \preceq rs=rt$ by the definition of $s$ and Proposition~\ref{prop:chain};
this contradicts that $r$ and $rt$ are not comparable.
Hence $s \not\in T_n$, and so $s$ is distinct from the transformations of Case~1.

We claim that $p$ is not in a cycle of $s$; this cycle would have to be
\begin{center}
 $p\stackrel{s }{\rightarrow} r \stackrel{s }{\rightarrow} rt \stackrel{s }{\rightarrow} \dots \stackrel{s }{\rightarrow} rt^{k-1}
\stackrel{s}{\rightarrow} p,
\text { that is, }
p\stackrel{s }{\rightarrow} r \stackrel{t }{\rightarrow} rt \stackrel{t }{\rightarrow} \dots \stackrel{t }{\rightarrow} rt^{k-1}
\stackrel{t}{\rightarrow} p,
$
\end{center}
\goodbreak
\noin
for some $k\ge 2$
because $r\neq p=pt$ and $rt\neq p$.
Since $p\prec r$  we have $p \prec rt$; but then we have a chain
$p\prec rt \prec \dots \prec rt^{k}=p$, contradicting Proposition~\ref{prop:chain}.

Since $p$ is not in a cycle of $s$, it follows that $s$ does not contain a cycle with states strictly ordered by $\prec$, as such a cycle would also be in $t$. So $s$ is distinct from the transformations of Case~2.

We claim there is a unique state $q$ such that (a) $0 \prec q \prec q s$, (b) $q s \not\preceq q s^2$.
First we show that $p$ satisfies these conditions: 
(a) holds because $ps=r$ and $p\prec r$;
(b) holds because $ps=r$, $ps^2=rt$ and $r$ and $rt$ are not comparable.
Now suppose that $q$ satisfies the two conditions, but $q \neq p$.
Note that $qs\neq p$, because $qs=p$ implies $qs=p\prec r=qs^2$, contradicting (b).
Since $q,q s \not\in \{0,p\}$, we have $q t = q s \not\preceq q s^2 = q t^2$. But Proposition~\ref{prop:chain} for $q \prec q t$ implies that $q t \preceq q t^2$---a contradiction.
Thus $p$ is the only state satisfying these conditions.

If $t' \neq t$ is another transformation satisfying the conditions of this case, we define $s'$ like $s$. Suppose that $s = f(t) = f(t') = s'$. Since both $s$ and $s'$ contain a unique state $p$ satisfying the two conditions above, we have $0 t = 0 t' = p$ and $p t = p t' = p$. 
Since the other states  are mapped by $s$ exactly as by $t$ and $t'$, we have $t = t'$.

\noindent
\hglue 15pt
$\bullet$ {\bf (b):} $t$ has no cycles and has a fixed point $r\neq p $.\\
Because $0\prec r$ by Remark~\ref{rem:left-ideals_xy2},
$0t\preceq rt$ by Proposition~\ref{prop:chain}.
If $r$ is a fixed point of $t$, then $p=0t \preceq rt=r$.
Since $r\neq p$, we have $p\prec r$.
Let $f(t)=s$, where $s$ is the transformation shown in Figure~\ref{fig:left-case3b} and defined by
\begin{center}
$0s=0, \quad q s = 0 \text{ for each fixed point } q\neq p,$

$q s = q t \text{ for the other states } q\in Q.$
\end{center}
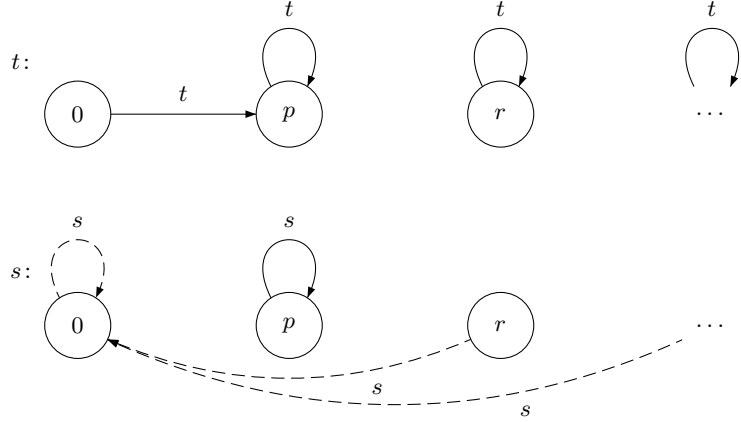
\begin{figure}[ht]
\unitlength 10pt
\begin{center}\begin{picture}(26,12)(0,0)
\gasset{Nh=2.5,Nw=2.5,Nmr=1.25,ELdist=0.5,loopdiam=2}
\node[Nframe=n](name)(0,12){$t\colon$}
\node(0)(2,10){0}
\node(p)(10,10){$p$}
\node(r)(18,10){$r$}
\node[Nframe=n](rdots)(26,10){$\dots$}
\drawedge(0,p){$t$}
\drawloop[loopangle=90](p){$t$}
\drawloop[loopangle=90](r){$t$}
\drawloop[loopangle=90](rdots){$t$}

\node[Nframe=n](name)(0,4){$s\colon$}
\node(0')(2,2){0}
\node(p')(10,2){$p$}
\node(r')(18,2){$r$}
\node[Nframe=n](rdots')(26,2){$\dots$}
\drawloop[loopangle=90,dash={.5 .25}{.25}](0'){$s$}
\drawloop[loopangle=90](p'){$s$}
\drawedge[curvedepth=2,ELpos=30,dash={.5 .25}{.25}](r',0'){$s$}
\drawedge[curvedepth=3,ELpos=30,dash={.5 .25}{.25}](rdots',0'){$s$}
\end{picture}\end{center}
\caption{Case~3(b) in the proof of Theorem~\ref{thm:left-ideals_upper-bound}.}
\label{fig:left-case3b}
\end{figure}

By Remark~\ref{rem:left-ideals_transformations}, $s \in S_n$.  Suppose that $s \in T_n$; because $p \prec r$, $ps=p$, $rs=0$, and 
$ps \preceq rs$ by Proposition~\ref{prop:chain}, we have $p \prec 0$, which is a contradiction.
Hence $s$ is not in $T_n$
and so is distinct from the transformations of Case~1. Also, $s$ maps at least one state other than $0$ to $0$, and so is distinct from the transformations of Case~2  and also from the transformations of Case~3(a).

If $t'\neq t$ is another transformation satisfying the conditions of this case, we define $s'$ like $s$.
Now suppose that $s = f(t) = f(t')=s'$.
There is only one fixed point of $s$ other than $0$ ($p s = p$), and only one fixed point of $s'$ other than 0 ($p' s' = p'$); hence $0 t= p= p'= 0 t'$. 
By the definition of $s$, for each state $q\neq 0$ such that $q s = 0$, we have $q t = q$.
Similarly, for each state $q\neq 0$ such that $qs'=0$, we have $q t' = q$. 
Hence $t$ and $t'$ agree on these states.
Since the remaining states are mapped by $s$ exactly as they are mapped by $t$ and $t'$, we have $t = t'$.
Thus we have proved that $t\neq t'$ implies $f(t)\neq f(t')$. 
\smallskip

\noindent
\hglue 15 pt
$\bullet$ {\bf (c):} $t$ has no cycles, has no fixed point $r\neq p$ and there is a state $r$ 
such that $p\prec r$ with $rt=p$.\\
Let $f(t) = s$, where $s$ is the transformation shown in Figure~\ref{fig:left-case3c} and defined by
\begin{center}
$0 s = 0, \hspace{.2cm} p s = r, \hspace{.2cm} q s = 0 \text{ for each $q \succ p$ such that } q t = p,$

$\hspace{.2cm} q s = q t \text{ for the other states } q \in Q.$
\end{center}
\begin{figure}[ht]
\unitlength 10pt
\begin{center}\begin{picture}(26,12)(0,0)
\gasset{Nh=2.5,Nw=2.5,Nmr=1.25,ELdist=0.5,loopdiam=2}
\node[Nframe=n](name)(0,12){$t\colon$}
\node(0)(2,10){0}
\node(p)(10,10){$p$}
\node(r)(18,10){$r$}
\node[Nframe=n](rdots)(26,10){$\dots$}
\drawedge(0,p){$t$}
\drawloop[loopangle=90](p){$t$}
\drawedge(r,p){$t$}
\drawedge[curvedepth=3](rdots,p){$t$}

\node[Nframe=n](name)(0,4){$s\colon$}
\node(0')(2,2){0}
\node(p')(10,2){$p$}
\node(r')(18,2){$r$}
\node[Nframe=n](rdots')(26,2){$\dots$}
\drawloop[loopangle=90,dash={.5 .25}{.25}](0'){$s$}
\drawedge[dash={.5 .25}{.25}](p',r'){$s$}
\drawedge[curvedepth=2,ELpos=30,dash={.5 .25}{.25}](r',0'){$s$}
\drawedge[curvedepth=3,ELpos=30,dash={.5 .25}{.25}](rdots',0'){$s$}
\end{picture}\end{center}
\caption{Case~3(c) in the proof of Theorem~\ref{thm:left-ideals_upper-bound}.}
\label{fig:left-case3c}
\end{figure}
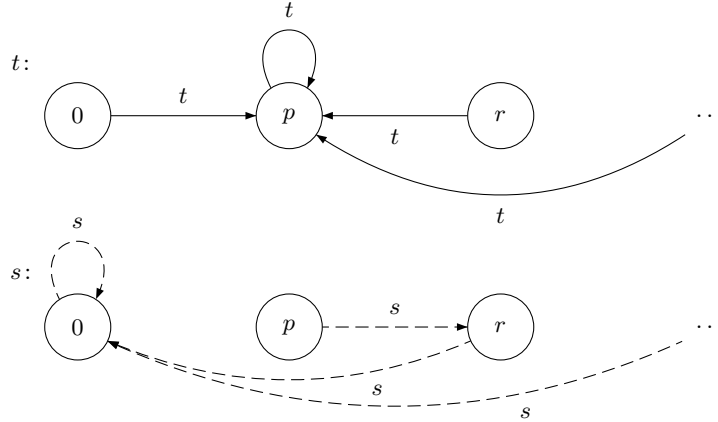

By Remark~\ref{rem:left-ideals_transformations}, $s \in S_n$. 
Suppose that $s \in T_n$; because $p \prec r$, $p s = r$, $r s = 0$, and $r= p s \preceq r s=0$ by Proposition~\ref{prop:chain}, we have $r \prec 0$---a contradiction. Hence $s \not\in T_n$ and $s$ is distinct from the transformations of Case~1.

Because $s$ maps at least one state other than $0$ to $0$ ($r s = 0$), it is distinct from the transformations of Case~2 and~3(a).
Also $s$ does not have a fixed point other than $0$, while the transformations of Case~3(b) have such a fixed point.

We claim that there is a unique state $q$ such that (a) $0 \prec q \prec q s$ and (b) $q s^2 = 0$. First we show that $p$ satisfies these conditions.
By assumption $0\prec p\prec r$ and  $rt=p$; also $rs=0$
by the definition of $s$.
Condition (a) holds because $0 \prec p \prec r = p s$, and (b) holds because $0 = r s = p s^2  $.

Now suppose that $0 \prec q \prec q s$, $q s^2 = 0$ and $q \neq p$.
Since $q s \neq 0$, we have $q s = q t$ by the definition of $s$. 
Because $q t$ has a $t$-predecessor, $p \preceq q t$ by Remark~\ref{rem:left-ideals_xy2}.
Also $q t = q s \neq p$, for $qs = p$ implies $0=qs^2=ps=r$---a contradiction.
Hence $p\prec q t$.
From $qt = qs$ and $q \prec q s$, we have $q \prec q t$.
Since $q s^2 = 0$ we have $(q t) s =0$ and so $(q t) t = p$,
 by the definition of $s$.
By Proposition~\ref{prop:chain}, from $q \prec q t $ we have $q t \preceq (q t) t = p$, contradicting $p \prec q t$.
So $q=p$.

If $t' \neq t$ is another transformation satisfying the conditions of this case, we define $s'$ like $s$. Suppose that $s = f(t) = t(t') = s'$. Since $s$ and $s'$ contain a unique state $p$ satisfying the two conditions above, we have $0 t = 0 t' = p$ and $p t = p t' = p$. Then $r$ and the states $q \succ p$ with $qt=p$ are determined by $p$, since they are precisely the states $q \succ p$ with $qs = 0$. Since the other states are mapped by $s$ exactly as by $t$ and $t'$, we have $t = t'$, and $f$ is again injective.
\smallskip

\hglue 5pt $\bullet$ {\bf All cases are covered:} \\ 
Now we need to ensure that any transformation $t$ fits in at least one case. It is clear that $t$ fits in Case~1 or~2 or~3. For Case~3, it is sufficient to show that if (i) $t \not\in S_n$ does not contain a fixed point $r\neq p$, and (ii) there is no state $r$ with $p \prec r$ and $r t = p$, then $t$ contains a cycle.

First, if there is no $r$ such that $p\prec r$, we claim that $t$ is the constant transformation $(Q\to p)$.
Consider any state $q \in Q$ such that $q t \neq p$. 
Then $p\prec qt$ by Remark~\ref{rem:left-ideals_xy2}, contradicting that there is no state $r$ such that $p\prec r$.

So let $r$ be some state such that $p\prec r$. Consider the sequence $r, rt, rt^2, \ldots$. By Remark~\ref{rem:left-ideals_xy2}, $p \preceq rt^i$ for all $i \ge 0$. 
If $rt^k = p$ for some $k \ge 1$, let $i$ be the smallest such $k$; we have $(rt^{i-1}) t = p$, contradicting (ii). Since $p$ is the only fixed point by (i), we have $rt^i \neq rt^{i-1}$.
Since there are finitely many states, $rt^i = rt^j$ for some $i$ and $j$ such that $0 \le i < j-1$, and so the states $rt^i,rt^{i+1}, \ldots, rt^j=rt^i$ form a cycle.
\smallskip

We have shown that for every transformation $t$ in $T_n$ there is a corresponding transformation $f(t)$ in $S_n$, and $f$ is injective. So $|T_n|\leq |S_n|=n^{n-1}+n-1$.
\qed
\end{proof}

Next we prove that $S_n$ is the only transition semigroup meeting the bound. It follows that minimal DFAs of left ideals with the maximal syntactic complexity have maximal-length chains of length $2$.

\begin{theorem}\label{thm:left-ideals_unique-maximal}
If $T_n$ has size $n^{n-1}+n-1$, then $T_n=S_n$. 
\end{theorem}
\begin{proof}
Consider a maximal-length chain of states strictly ordered by $\prec$ in $\cD_n$. If its length is 2, then by Lemma~\ref{lem:chain2}, $T_n$ is a subsemigroup of $S_n$. Thus only $T_n = S_n$ reaches the bound in this case.

Assume now that the length of a maximal-length chain is at least 3. 
Then there are states $p$ and $r$ such that $0 \prec p \prec r$.
Let $R=\{q \mid p\prec q\}$, and let $X=Q\setminus (R\cup \{0,p\})$.
We shall show that there exists a transformation $s$ that is in $S_n$ but not in $f(T_n)$.
To define $s$ we use  the constant transformation $u=(Q\to p)$ as an auxiliary transformation.
Let $0s=0$, $ps=r$, $rs=0$ for all $r\in R$, and $qs = qu=p$ for $q\in X$; these are precisely the rules  we used in Case 3(c) in the proof of Theorem~\ref{thm:left-ideals_upper-bound}.
By Remark~\ref{rem:left-ideals_transformations}, $s \in S_n$.

It remains to be shown that there is no transformation $t\in T_n$ such that $s=f(t)$.
The proof that $s$ is different from the transformations $f(t)$ of Cases~1, 2, 3(a) and 3(b) is exactly the same as the corresponding proof in Case~3(c) following the definition of $s$.

It remains to verify that there is no $u' \in T_n$\ in Case~3(c) such that $f(u')=s$.
Suppose there is such a $u'$.  
Recall that states $p$ and $r$ satisfying $0 \prec p \prec r$ have been fixed by assumption. 
By the definition of $s$, state  $p$ satisfies the conditions (a) $0 \prec p \prec ps$ and (b) $ps^2 = 0$.
We claim that $p$ is the only state satisfying these conditions.
Indeed, if $q \neq p$ then either $qs = 0$, $q \not\prec qs=0$ and (a) is violated, or
$qs = p$,  $qs^2 = ps = r \neq 0$ and (b) is violated.
This observation is used in the proof of Case~3(c) to prove the claim below.

Both $u$ and $u'$ satisfy the conditions of Case~3(c), except that $u$ fails the condition $u\not\in S_n$. 
However, that latter condition is not used in the proof that 
if $u \neq u'$ and $u'$ satisfy the other conditions of Case~3(c), then $s' \neq s$, where $s'$ is the transformation obtained from $u'$ by the rules of $s$. Thus $s$ is also different from the transformations in $f(T_n)$ from Case~3(c).

Because $s \not\in f(T_n)$, $s \in S_n$ and $f(T_n) \subseteq S_n$,
the bound $n^{n-1}+n-1$ cannot be reached if the length of the maximal-length chains is not 2.
\qed
\end{proof}

\section{Two-Sided Ideals}

If a language $L$ is a right ideal, then $L=L\Sig^*$ and $L$ has exactly one final quotient, namely $\Sig^*$;
hence this also holds for two-sided ideals.
For $n\ge 3$, in a two-sided ideal every maximal chain is of length at least 3: it starts with $L$, every quotient contains $L$ and is contained in $\Sig^*$.

It was proved in~\cite[Theorem 6, p. 125]{BrYe11} that the transition semigroup of the following DFA of a two-sided ideal meets the bound $n^{n-2} + (n-2) 2^{n-2} +1$.

\begin{definition}[Witness: Two-Sided Ideals]
\label{def:wit2}
For  $n\ge 4$, define the  DFA
$\cW_n =(Q,\Sig_\cW,\delta_\cW, 0,\{n-1\})$,
where $Q=\{0,\ldots,n-1\}$, $\Sig_\cW=\{a,b,c,d,e,f\}$,
and $\delta_\cW$ is defined by $a\colon (1,2,\ldots,n-2)$,
$b\colon (1,2)$,
$c\colon (n-2\to 1)$,
$d\colon (n-2\to 0)$,
for $q=0,\ldots,n-2$, $\delta(q,e)=1$ and $\delta(n-1,e)=n-1$,
and $f\colon (1\to n-1)$.
For $n=4$, inputs $a$ and $b$ coincide, and we can use $\Sig_\cW=\{b,c,d,e,f\}$.
\end{definition}

\begin{remark}
\label{rem:2sided}
If $n=1$, the only two-sided ideal is $\Sig^*$, its syntactic complexity is 1, and the  bound above is not tight. 
If $n=2$, each two-sided ideal is of the form $L=\Sig^*\Gamma\Sig^*$, where $\emp\subsetneq \Gamma\subseteq \Sigma$, its syntactic complexity is $2$, and the bound is tight.
If $n=3$, there are eight transformations that are initially aperiodic and such that $(n-1)t = t$ (the property of a right-ideal transformation).
We have verified that the DFA having all eight or any seven of the eight transformations is not a two-sided ideal. Hence $6$ is an upper bound, and we know from~\cite{BrYe11} that the transformations
$[1,2,2]$, $[0,0,2]$, and $[0,1,2]$ generate a 6-element semigroup.
From now on we may assume that $n\ge 4$.
\end{remark}

We consider a minimal DFA $\cD_n=(Q, \Sigma_\cD, \delta_\cD, 0,\{n-1\})$, where $Q=\{0,\ldots,n-1\}$, of an arbitrary two-sided ideal with $n$ quotients, and let $T_n$ be the transition semigroup of $\cD_n$. 
We also deal with the witness DFA $\cW_n =(Q,\Sig_\cW,\delta_\cW,0,\{n-1\})$ of Definition~\ref{def:wit2} with transition semigroup  $S_n$.

\begin{remark}\label{rem:2sided_transformations}
In $\cW_n$, the transformations induced by $a$, $b$, and $c$ restricted to $Q \setminus \{0,n-1\}$ generate all the transformations of the states $1,\ldots,n-2$. Together with the transformations of $d$ and $f$, they generate all transformations of $Q$ that fix $0$ and $n-1$.
For any subset $S \subseteq \{1,\ldots,n-2\}$, there is a transformation---induced by a word $w_S$, say---that maps $S$ to $n-1$ and fixes $Q \setminus S$. Then the words of the form $w_S e a^i$, for $i \in \{0,\ldots,n-3\}$, induce all transformations that maps $S \cup \{n-1\}$ to $n-1$ and $Q \setminus (S \cup \{n-1\})$ to $i+1$. In $\cW_n$, there is also the constant transformation $ef\colon (Q \to n-1)$.
\end{remark}

\begin{lemma}\label{lem:chain3}
If $n \ge 4$ and a maximal-length chain in $\cD_n$ strictly ordered by $\prec$ has length 3, then $|T_n|\le n^{n-2}+(n-2)2^{n-2}+1$, and $T_n$ is a subsemigroup of $S_n$.
\end{lemma}
\begin{proof}
Consider an arbitrary transformation $t\in T_n$; then $(n-1)t=n-1$. 
If $0t=0$, then any state not in $\{0,n-1\}$ can possibly be mapped by $t$  to any one of the $n$ states; hence there are at most $n^{n-2}$ such transformations.

If $0t\neq 0$, then
$0\prec 0t$. Consider any state $q\not \in \{0,0t\}$; since $\cD_n$ is minimal, $q$ must be reachable from 0 by some transformation $s$, that is, $q=0s$.
If $0st \not\in \{0t,n-1\}$, then $0t\prec 0st$ by Remark~\ref{rem:left-ideals_xy2}.
But then we have the chain $0\prec 0t\prec 0st\prec n-1$ of length 4, contradicting our assumption.
Hence we must have either $0st=0t$, or $0st=n-1$.
For a fixed $0t$, a subset of the states in $Q\setminus\{0,n-1\}$ can be mapped to $0t$ and the remaining states in $Q\setminus\{0,n-1\}$ to $n-1$, thus giving $2^{n-2}$ transformations.
Since there are $n-2$ possibilities for $0t$, we obtain the second part of the bound.
Finally, all states can be mapped to $n-1$.

By Remark~\ref{rem:2sided_transformations} all of the above-mentioned transformations are in $S_n$.
\qed
\end{proof}

\begin{theorem}[Two-Sided Ideals, Factor-Closed Languages]\label{thm:2sided-ideals_upper-bound}
If $L$ is a two-sided ideal or a factor-closed language with $n\ge 4$ quotients, then its syntactic complexity is less than or equal to 
$n^{n-2} + (n-2) 2^{n-2} +1$.
\end{theorem}
\begin{proof}
It suffices to prove the result for two-sided ideals.
As we did for left ideals,
we show that $|T_n| \le |S_n|$, by constructing an injective function $f\colon T_n \to S_n$.

We have $q \preceq n-1$ for any $q \in Q$, and $n-1$ is a fixed point of every transformation in $T_n$ and $S_n$.

For a transformation $t \in T_n$, consider the following cases:
\smallskip

\noin
\hglue 15 pt
{\bf Case 1:} $t \in S_n$. \\
The proof is the same as that of Case 1 of Theorem~\ref{thm:left-ideals_upper-bound}.
\smallskip

\noin
\hglue 15 pt
{\bf Case 2:} $t \not\in S_n$, and $0 t^2 \neq 0t$. \\
Let $0t=p \prec \dots \prec p t^k = p t^{k+1}$ be the chain defined from $p$.
\smallskip

\noin
\hglue 15 pt
$\bullet$ {\bf (a):} $p t^k \neq n-1$. \\
The proof is the same as that of Case 2 of Theorem~\ref{thm:left-ideals_upper-bound}.\smallskip

\noin
\hglue 15 pt
$\bullet$ {\bf (b):} $p t^k = n-1$ and $k \ge 2$.\\
Let $f(t)=s$, where $s$ is the transformation shown in Figure~\ref{fig:2sided-case2b} and defined by
\begin{center}
$0 s = 0, \quad p t^i s = p t^{i-1} \text{ for }1 \le i \le k-1 , \quad p s = n-1,$

$q s = q t \text{ for the other states } q\in Q.$
\end{center}
\begin{figure}[ht]
\unitlength 10pt
\begin{center}\begin{picture}(27,12)(0,0)
\gasset{Nh=2.5,Nw=2.5,Nmr=1.25,ELdist=0.5,loopdiam=2}
\node[Nframe=n](name)(0,12){$t\colon$}
\node(0)(2,10){0}
\node(p)(7,10){$p$}
\node(pt)(12,10){$pt$}
\node[Nframe=n](pdots)(17,10){$\dots$}
\node(pt^{k-1})(22,10){$pt^{k-1}$}
\node(n-1)(27,10){$n-1$}
\drawedge(0,p){$t$}
\drawedge(p,pt){$t$}
\drawedge(pt,pdots){$t$}
\drawedge(pdots,pt^{k-1}){$t$}
\drawedge(pt^{k-1},n-1){$t$}
\drawloop[loopangle=90](n-1){$t$}

\node[Nframe=n](name)(0,4){$s\colon$}
\node(0')(2,2){0}
\node(p')(7,2){$p$}
\node(pt')(12,2){$pt$}
\node[Nframe=n](pdots')(17,2){$\dots$}
\node(pt^{k-1}')(22,2){$pt^{k-1}$}
\node(n-1')(27,2){$n-1$}
\drawloop[loopangle=90,dash={.5 .25}{.25}](0'){$s$}
\drawedge[ELdist=-1,dash={.5 .25}{.25}](pt',p'){$s$}
\drawedge[ELdist=-1,dash={.5 .25}{.25}](pdots',pt'){$s$}
\drawedge[ELdist=-1,dash={.5 .25}{.25}](pt^{k-1}',pdots'){$s$}
\drawedge[curvedepth=-2.5,ELdist=-1,dash={.5 .25}{.25}](p',n-1'){$s$}
\drawloop[loopangle=90](n-1'){$s$}
\end{picture}\end{center}
\caption{Case~2(b) in the proof of Theorem~\ref{thm:2sided-ideals_upper-bound}.}
\label{fig:2sided-case2b}
\end{figure}
By Remark~\ref{rem:2sided_transformations}, $s \in S_n$. Note that $s$ contains the cycle $(p, pt)$ where $p t \succ p$,  $p t s = p$ and $p s = n-1$. By Proposition~\ref{prop:chain}, 
$pts\succeq ps$, that is, 
$p \succeq  n-1$, which contradicts the fact that $p \neq n-1$, since $pt \neq p$. Thus  $s$ is not in $T_n$, and so it is different from the transformations of Case~1.

Observe that $s$ does not have a cycle with states strictly ordered by $\prec$, since no state from $\{0, p, pt, \ldots, pt^{k-1} \}$ can be in a cycle, and $t$ cannot have such a cycle. Hence $s$ is different from the transformations of Case~2(a).

In $s$, there is a unique state $q$ such that $q s = n-1$ and for which there exists a state 
$r$ such that $r \succ q$ and $r s = q$, and that this state $q$ must be $p$.
Indeed, if $q \neq p$, then $q t = q s = n-1$ by the definition of $s$.
From $r \succ q$, we have $rt\succeq qt=n-1$; hence $rs=rt=n-1$ and $rt\neq q$---a contradiction. Hence $q=p$.

By a similar argument, we show that there exists a unique state $q$ such that $q \succ p$, and $q s = p$, and that this state $q$ must be $pt$. If $q \neq pt$ then $qs = qt$. But $q \succ qt$ and $p=qt \succeq qt^2=pt$ contradicts that $p \prec pt$. Continuing in this way for $pt^2,\ldots,pt^{k-1}$ we show that there is a unique chain
$pt^{k-1} \stackrel{s }{\rightarrow} \dots \stackrel{s }{\rightarrow} pt \stackrel{s }{\rightarrow} p$.

If $t'\neq t$ is another transformation satisfying the conditions of this case, we define $s'$ like $s$.
Now suppose that $s = f(t) = f(t') = s'$.
Since we have a unique state $p$ such that $p s = n-1$ for which there exists a state $r$ such that $r \succ p$ and $r s = p$, we have $0 t = 0 t' = p$.
Also the chain of states $p,pt,pt^2,\dots,pt^{k-1}$ is unique in $s$ and $s'$ as we have shown above; so $p t^{i} = p {t'}^{i}$ for $i=1,\dots,k-1$. Since the other states are mapped by $s$ exactly as by $t$ and $t'$, we have $t = t'$.
\smallskip

\noin
\hglue 15 pt
$\bullet$ {\bf (c):} $p t = n-1$.\\
Let $P=\{0,p,n-1\}$. Since $n\ge 4$, there must be a state $r\not\in P$. 
If $p\prec r$ for all $r\not\in P$, then $n-1=pt\preceq rt$; hence $rt=n-1$ for all such $r$,
and $q t \in \{p,n-1\}$ for all $q \in Q$.
By Remark~\ref{rem:2sided_transformations}, there is a transformation in $S_n$ that maps $S \cup \{n-1\}$ to $n-1$, and $Q \setminus (S \cup \{n-1\})$ to $p$ for any $S \subseteq \{1,\ldots,n-2\}$. Thus $t \in S_n$---a contradiction.

In view of the above, there must exist a state $r\not\in P$ such that $p\not\preceq r$.
By Remark~\ref{rem:left-ideals_xy2}, we have $p\preceq rt$ and of course $rt\preceq n-1$. If $rt$ is $p$ or $n-1$ for all $r \not\in P$, we again have the situation described above, showing that $t \in S_n$.
Hence there must exist an $r\not\in P$ such that $p\not\preceq r$ and $p \prec r t \prec n-1$.

Also we claim that $t$ does not have a cycle. Indeed, if $p \preceq q$, then $q$ is mapped to $n-1$; if $p \not\preceq q$, then $q$ is mapped to a state $q t \succeq p$ and again $q$ cannot be in a cycle since the chain starting with $q$ ends in $n-1$.

Let $f(t)=s$, where $s$ is the transformation shown in Figure~\ref{fig:2sided-case2c} and defined by
\begin{center}
$0 s = 0, \quad p s = r t, \quad (r t) s = p, \quad r s = 0,$

$q s = q t \text{ for the other states } q\in Q.$
\end{center}
\begin{figure}[ht]
\unitlength 10pt
\begin{center}\begin{picture}(18,14)(0,0)
\gasset{Nh=2.5,Nw=2.5,Nmr=1.25,ELdist=0.5,loopdiam=2}
\node[Nframe=n](name)(0,14){$t\colon$}
\node(0)(2,10){0}
\node(p)(10,10){$p$}
\node(n-1)(18,10){$n-1$}
\node(r)(6,14){$r$}
\node(rt)(14,14){$rt$}
\drawedge(0,p){$t$}
\drawedge(p,n-1){$t$}
\drawloop[loopangle=90](n-1){$t$}
\drawedge(r,rt){$t$}
\drawedge(rt,n-1){$t$}

\node[Nframe=n](name)(0,6){$s\colon$}
\node(0')(2,2){0}
\node(p')(10,2){$p$}
\node(n-1')(18,2){$n-1$}
\node(r')(6,6){$r$}
\node(rt')(14,6){$rt$}
\drawloop[loopangle=90,dash={.5 .25}{.25}](0'){$s$}
\drawedge[curvedepth=1,dash={.5 .25}{.25}](p',rt'){$s$}
\drawedge[curvedepth=1,dash={.5 .25}{.25}](rt',p'){$s$}
\drawedge[dash={.5 .25}{.25}](r',0'){$s$}
\drawloop[loopangle=90](n-1'){$s$}
\end{picture}\end{center}
\caption{Case~2(c) in the proof of Theorem~\ref{thm:2sided-ideals_upper-bound}.}
\label{fig:2sided-case2c}
\end{figure}

Since $s$ fixes both 0 and $n-1$, it is in $S_n$ by Remark~\ref{rem:2sided_transformations}. But $s$ is not in $T_n$, as we have the cycle $(p, r t)$ with $p \prec r t$. So $s$ is different from the transformations of Case~1.
Since $s$ maps a state other than $0$ to $0$, it is different from the transformations of Cases~2(a) and 2(b).

Observe that $t$ does not map any state to 0. Consequently, in $s$ there is the unique state $r\neq 0$ mapped to $0$. Also, as $t$ does not contain a cycle, the only cycle in $s$ must be $(p, r t)$.

If $t'\neq t$ is another transformation satisfying the conditions of this case, we define $s'$ like $s$. Now suppose that $s = f(t) = f(t') = s'$.
Because both $s$ and $s'$ have the unique non-fixed point $r$ mapped to $0$, $r = r'$. Also $s$ and $s'$ contain the unique cycle $(p, r t)$, $p \prec r t$. Thus $p = p'$, $p t = p t' = n-1$ and $r t = r t'$. It follows that $0 t = 0 t' = p$. Because $p \prec r t = r t'$, we have $(r t) t = (r t) t' = n-1$. The other states are mapped by $s$ exactly as by $t$ and $t'$, and so $t = t'$.
\smallskip

\noin
\hglue 15 pt
{\bf Case 3:} $t \not\in S_n$, $0 t = p \neq 0$ and $p t = p$.

\noin
\hglue 15 pt
$\bullet$ {\bf (a):} $t$ has a cycle.\\
The proof is analogous to that of Case~3(a) in Theorem~\ref{thm:left-ideals_upper-bound}, 
but we need to ensure that $s$ is different from the $s$ of Cases~2(b) and 2(c).

Here there is the state $r$ such that $r \prec r s$, and $r s$ and $r s^2$ are not comparable under $\preceq$.
Consider a transformation $t'$ that fits in Case~2(b). Then in $s'$ every state $q = pt^i$ for $0 \le i \le k-1$, and $q = 0$, is mapped to a state comparable with $q$ under $\preceq$, and the other states are mapped as in $t'$. Since $t' \in T_n$ cannot contain a state $r'$ such that $r' \prec r' t$ and  $r' t$ and $r' t^2$ are not comparable under $\preceq$, it follows that $s'$ also does not contain such a state. Thus $s \neq s'$.

For a distinction from the transformations of Case~2(c) observe that $s$ does not map to 0 any state other than 0.\smallskip

\noin
\hglue 15 pt
$\bullet$ {\bf (b):} $t$ has no cycles and has a fixed point $r \not\in \{p,n-1\}$.\\
The proof is analogous to that of Case~3(b) in Theorem~\ref{thm:left-ideals_upper-bound}, 
but we need to ensure that $s$ is different from the $s$ of Cases~2(b) and 2(c).

Since $s$ maps to 0 a state other than 0, this case is distinct from Case~2(b).
Because $t$ does not have a cycle, and no state $q$ mapped to $0$ can be in a cycle in $s$, it follows that $s$ does not have a cycle. Thus $s$ is different from the transformations of Case~2(c).
\smallskip

\noin
\hglue 15 pt
$\bullet$ {\bf (c):} $t$ has no cycles and no fixed point $r \not\in \{p,n-1\}$, but has a state $r \succ p$ mapped to $p$.\\
\noin The proof is analogous to that of Case~3(c) in Theorem~\ref{thm:left-ideals_upper-bound},
but we need to ensure that $s$ is different from the $s$ of Cases~2(b) and 2(c).

As before, since $s$ maps to 0 a state other than $0$, this case is distinct from Case~2(b).
In $s$, $0$ cannot be in a cycle, no state $q \succ p$ mapped to $0$ can be in a cycle and $p$ cannot be in a cycle as $ps = r$ and $rs = 0$. Since the other states are mapped as in $t$, $s$ does not have a cycle. Thus $s$ is different from the transformations of Case~2(c).
\smallskip

\noin
\hglue 15 pt
$\bullet$ {\bf (d):} $t$ has no cycles, no fixed point $r \not\in \{p,n-1\}$, and no state $r \succ p$ mapped to $p$, but has a state $r$ such that $p \prec r \prec n-1$, mapped to $n-1$.\\
Let $f(t)=s$, where $s$ is the transformation shown in Figure~\ref{fig:2sided-case3d} and defined by
\begin{center}
$0 s = 0, \quad q s = q \text{ for states } q \text{ such that } q t = n-1, \quad p s = n-1$
$qs = qt \text{ for the other states } q\in Q.$
\end{center}
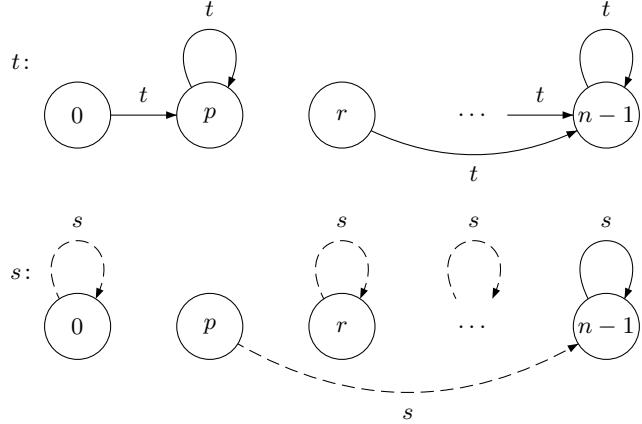
\begin{figure}[ht]
\unitlength 10pt
\begin{center}\begin{picture}(22,12)(0,0)
\gasset{Nh=2.5,Nw=2.5,Nmr=1.25,ELdist=0.5,loopdiam=2}
\node[Nframe=n](name)(0,12){$t\colon$}
\node(0)(2,10){0}
\node(p)(7,10){$p$}
\node(r)(12,10){$r$}
\node[Nframe=n](rdots)(17,10){$\dots$}
\node(n-1)(22,10){$n-1$}
\drawedge(0,p){$t$}
\drawloop[loopangle=90](p){$t$}
\drawloop[loopangle=90](n-1){$t$}
\drawedge[curvedepth=-1.5,ELdist=-1](r,n-1){$t$}
\drawedge(rdots,n-1){$t$}

\node[Nframe=n](name)(0,4){$s\colon$}
\node(0')(2,2){0}
\node(p')(7,2){$p$}
\node(n-1')(22,2){$n-1$}
\node(r')(12,2){$r$}
\node[Nframe=n](rdots')(17,2){$\dots$}
\drawloop[loopangle=90,dash={.5 .25}{.25}](0'){$s$}
\drawloop[loopangle=90,dash={.5 .25}{.25}](r'){$s$}
\drawloop[loopangle=90,dash={.5 .25}{.25}](rdots'){$s$}
\drawedge[curvedepth=-2.5,ELdist=-1,dash={.5 .25}{.25}](p',n-1'){$s$}
\drawloop[loopangle=90](n-1'){$s$}
\end{picture}\end{center}
\caption{Case~3(d) in the proof of Theorem~\ref{thm:2sided-ideals_upper-bound}.}
\label{fig:2sided-case3d}
\end{figure}

By Remark~\ref{rem:2sided_transformations}, $s \in S_n$. However, $s$ is not in $T_n$, as we have a fixed point $r$ such that $p \prec r \prec n-1$ and $p s = n-1$. So Proposition~\ref{prop:chain} yields $n-1 = p s \preceq r s = r$---a contradiction. Thus $s$ is different from the transformations of Case~1.

Transformation $s$ does not have any cycles, as $t$ does not have one in this case and fixed points $q$ and $p$ cannot be in a cycle. So $s$ is different from the transformations of Cases~2(a) and~3(a).
Also, since  $p$ is the unique state mapped to $n-1$ and there is no state $r \succ p$ mapped to $p$,  $s$ is different from the transformations of Case~2(b).
For a distinction from the transformations of Cases~2(c), 3(b) and 3(c), observe that $s$ does not map to 0 any state other than $0$.

If $t'\neq t$ is another transformation satisfying the conditions of this case, we define $s'$ like $s$.
Now suppose that $s = f(t) = f(t') = s'$.
Observe that $t$ does not have a fixed point other than $n-1$. So for every fixed point $q \not\in \{0,n-1\}$ of $s$ we have $q t = q t' = n-1$. Also, since $p$ is the unique state mapped to $n-1$ in $s$, $0 t = 0 t' = p$ and $p t = p t' = p$. The other states are mapped by $s$ as by $t$ and $t'$; so $t = t'$.
\smallskip

\hglue 5pt $\bullet$ {\bf All cases are covered:} \\ 
We need to ensure that any transformation $t$ fits in at least one case. It is clear that $t$ fits in Case~1 or 2 or 3. Any transformation from Case~2 fits in Case~2(a) or 2(b) or 2(c).
For Case~3, it is sufficient to show that if (i) $t \not\in S_n$ does not contain a fixed point $r\not\in \{p,n-1\}$, and (ii) there is no state $r$, $p \prec r \prec n-1$, mapped to $p$ or $n-1$, then $t$ has a cycle.

If there is no state $r$ such that $p \prec r \prec n-1$, then $q t \in \{p,n-1\}$ for any $q \in Q$, since $q t \succeq p$; by Remark~\ref{rem:2sided_transformations}, $t \in S_n$---a contradiction.

So let $r$ be some state such that  $p \prec r \prec n-1$. Consider the sequence $r,rt,rt^2,\ldots$. By Remark~\ref{rem:left-ideals_xy2}, $p \preceq rt^i$ for all $i \ge 0$. 
If $rt^k \in \{p,n-1\}$ for some $k \ge 1$, then let $i$ be the smallest such $k$. Then we have $(rt^{i-1}) t \in p$, contradicting (ii). Since $p$ and $n-1$ are the only fixed points by (i), we have $rt^i \neq rt^{i-1}$. Since there are finitely many states, $rt^i = rt^j$ for some $i$ and $j$ such that $0 \le i < j-1$, and so the states $rt^i,rt^{i+1}\ldots,rt^j=rt^i$ form a cycle.
\qed
\end{proof}

\begin{theorem}\label{thm:2sided-ideals_unique-maximal}
If $T_n$ has size $n^{n-2} + (n-2) 2^{n-2} +1$, then $T_n=S_n$.
\end{theorem}
\begin{proof}
The proof is very similar to that of Theorem~\ref{thm:left-ideals_unique-maximal}.

Consider a maximal-length chain of states strictly ordered by $\prec$ in $\cD_n$. If its length is 3, then by Lemma~\ref{lem:chain3} $T_n$ is a subsemigroup of $S_n$. Thus only $T_n = S_n$ reaches the bound. 

If there is a chain of length 4, 
then there are states $p$ and $r$ such that $0 \prec p \prec r \prec n-1$. Let $f$ be the injective function from Theorem~\ref{thm:2sided-ideals_upper-bound}. Consider the transformation $u$ that maps $Q \setminus \{n-1\}$ to $p$ and fixes $n-1$.
Let $s$ be defined from $u$ in Case~3(c) of the proof of Theorem~\ref{thm:2sided-ideals_upper-bound}. The rest of the proof follows the proof of Theorem~\ref{thm:left-ideals_unique-maximal} with Case~3(d) of Theorem~\ref{thm:2sided-ideals_upper-bound} added.
\qed
\end{proof}

\providecommand{\noopsort}[1]{}

%

\end{document}